%% file: main.tex
\documentclass[11pt]{article}
\usepackage[margin=1in]{geometry}
\usepackage{complexity}
\usepackage{mathtools,amsmath,amssymb}
\usepackage{xspace}
\RequirePackage[colorlinks=true]{hyperref}
\hypersetup{
  linkcolor=[rgb]{0,0,0.5},
  citecolor=[rgb]{0, 0.5, 0},
  urlcolor=[rgb]{0.5, 0, 0}
}
\usepackage{dsfont}

\input{thmmacros}
\input{customurlbst/bibmacros}

	\renewcommand{\vec}[1]{{\mathbf{#1}}}

	\makeatletter
	\newcommand{\va}{{\vec{a}}\@ifnextchar{^}{\!\:}{}}
	\newcommand{\vb}{{\vec{b}}\@ifnextchar{^}{\!\:}{}}
	\newcommand{\vc}{{\vec{c}}\@ifnextchar{^}{\!\:}{}}
	\newcommand{\vd}{{\vec{d}}\@ifnextchar{^}{\!\:}{}}
	\newcommand{\ve}{{\vec{e}}\@ifnextchar{^}{\!\:}{}}
	\newcommand{\vy}{{\vec{y}}\@ifnextchar{^}{\!\:}{}}
	\newcommand{\vs}{{\vec{s}}\@ifnextchar{^}{\!\:}{}}
	\newcommand{\vt}{{\vec{t}}\@ifnextchar{^}{\!\:}{}}
	\newcommand{\vx}{{\vec{x}}\@ifnextchar{^}{}{}}		%\vec{x} seems fine already
	\newcommand{\vz}{{\vec{z}}\@ifnextchar{^}{\!\:}{}}

	\newcommand{\vY}{{\vec{Y}}\@ifnextchar{^}{\!\:}{}}
	\newcommand{\vX}{{\vec{X}}\@ifnextchar{^}{}{}}		%\vec{x} seems fine already
	\newcommand{\vZ}{{\vec{Z}}\@ifnextchar{^}{\!\:}{}}
	\newcommand{\vG}{{\vec{G}}\@ifnextchar{^}{\!\:}{}}
	
	\makeatother

\newcommand{\cF}{{\mathcal{F}}}

\newcommand{\cL}{{\mathcal{L}}}
\newcommand{\cU}{{\mathcal{U}}}
\newcommand{\cW}{{\mathcal{W}}}

\newcommand{\F}{\mathbb{F}}
\newcommand{\N}{\mathbb{N}}
\newcommand{\Z}{\mathbb{Z}}
\renewcommand{\E}{\mathbb{E}}

\newcommand{\ind}{\mathds{1}}
\newcommand{\set}[1]{\left\{#1\right\}}
\newcommand{\setdef}[2]{\set{#1 : #2}}
\newcommand{\abs}[1]{\left|#1\right|}
\newcommand{\floor}[1]{\left\lfloor#1\right\rfloor}
\newcommand{\ip}[1]{\left\langle#1\right\rangle}
\newcommand{\me}{\mathrm{e}}
\DeclareMathOperator{\rank}{rank}

\def\epsilon{\varepsilon}

\date{}

\title{Unbalancing Sets and an Almost Quadratic Lower Bound 
for Syntactically Multilinear Arithmetic Circuits}
\author{
Noga Alon\thanks{Sackler School of Mathematics
and Blavatnik School of
Computer Science, Tel Aviv University, Tel Aviv 6997801, Israel
and CMSA, Harvard University, Cambridge, MA 02138, USA.
Email: \texttt{nogaa@tau.ac.il}.  Research supported in part by
an ISF grant and by a GIF grant.
}
\and 
Mrinal Kumar\thanks{Center for Mathematical Sciences and Applications, Harvard University, Cambridge, Massachusetts, USA. Email: \texttt{mrinalkumar08@gmail.com}. Part of this work was done while visiting Tel Aviv University.}
\and%
Ben Lee Volk\thanks{Blavatnik School of Computer Science, Tel Aviv University, Tel Aviv, Israel, Email: \texttt{benleevolk@gmail.com}. The research leading to these results has received funding from the Israel Science Foundation (grant number 552/16).}
}
\begin{document}
\maketitle

\begin{abstract}
We prove a lower bound of $\Omega(n^2/\log^2 n)$ on the size of any
syntactically multilinear arithmetic circuit computing some explicit
multilinear polynomial $f(x_1, \ldots, x_n)$. Our approach expands and
improves upon a result of Raz, Shpilka and Yehudayoff (\cite{RSY08}),
who proved a lower bound of $\Omega(n^{4/3}/\log^2 n)$ for the same
polynomial. Our improvement follows from an asymptotically optimal
lower bound for a generalized version of Galvin's problem in extremal
set theory.
\end{abstract}

\thispagestyle{empty}
\newpage
\pagenumbering{arabic}

\section{Introduction}
\label{sec:intro}
An arithmetic circuit is one of the most natural and standard
computational models for computing multivariate polynomials. 
Such circuits provide
a succinct representation of multivariate polynomials, and in some sense,
they can be thought of as algebraic analogs of boolean circuits. Formally,
an arithmetic circuit over a field $\F$ and a set of variables $X =
\{x_1, x_2, \ldots, x_n\}$ is a directed acyclic graph in which every
vertex has in-degree either zero or two. The vertices of in-degree zero
(called \emph{leaves}) are labeled by variables in $X$ or elements of
$\F$, and the vertices of in-degree two are labeled by either $+$ (called
\emph{sum} gates) or $\times$ (called \emph{product} gates). A circuit
can have one or more vertices of out degree zero, known as the output
gates. The polynomial computed by a vertex in any\footnote{Throughout
this paper, we will use the terms gates and vertices interchangeably.}
given circuit is naturally defined in an inductive way: a leaf computes
the polynomial which is equal to its label. A sum gate computes the
polynomial which is the sum of the polynomials computed at its children
and a product gate computes the polynomial which is the product of the
polynomials at its children. The polynomials computed by a circuit are
the polynomials computed by its output gates. The size of an arithmetic
circuit is the number of vertices in it.

It is not hard to show (see, e.g., \cite{ckw11}) that a random
polynomial of degree $d = \poly(n)$ in $n$ variables cannot be computed
by an arithmetic circuit of size $\poly(n)$ with overwhelmingly high
probability. A fundamental problem in this area of research  is to
prove a similar super-polynomial lower bound for an \emph{explicit}
polynomial family. Unfortunately,  the problem continues to remain
wide open and the current best lower bound known for general arithmetic
circuits\footnote{In the rest of the paper, when we say a lower bound, we
always mean it for an explicit polynomial family.} is an $\Omega(n\log n)$
lower bound due to Strassen~\cite{Str73} and Baur and Strassen~\cite{BS83}
from more than three decades ago. The absence of substantial progress on
this general question has led to focus on the question of proving better
lower bounds for restricted and more structured subclasses of arithmetic
circuits. Arithmetic formulas~\cite{k85}, non-commutative arithmetic
circuits~\cite{nis91}, algebraic branching programs~\cite{k17}, and
low depth arithmetic circuits~\cite{nw1997, grigoriev98, gr00, Raz10a,
gkks13, FLMS13, KLSS, KS14, KumarSapth15} are some such subclasses
which have been studied from this perspective. For an overview of the
definition of these models and the state of art for lower bounds for them,
we refer the reader to the surveys of Shpilka and Yehudayoff~\cite{sy}
and Saptharishi~\cite{github}.

Several of the most important polynomials in algebraic complexity
and in mathematics in general are multilinear. Notable examples
include the determinant, the permanent, and the elementary symmetric
polynomials. Therefore, one subclass which has received a lot of attention
in the last two decades and will be the focus of this paper is the class
of \emph{multilinear} arithmetic circuits.

\subsection{Multilinear arithmetic circuits}
For an arithmetic circuit $\Psi$ and a vertex $v$ in $\Psi$, we denote
by $X_v$ the set of variables $x_i$ such that there is a directed path
from a leaf labeled by $x_i$ to $v$; in this case, we also say that $v$
\emph{depends} on $x_i$\footnote{We remark that this is a syntactic
notion of dependency, since it is possible that every monomial with
$x_i$ might get canceled in the intermediate computation and might not
eventually appear in the polynomial computed at $v$.}. A polynomial $P$
is said to be multilinear if the individual degree of every variable in
$P$ is at most one.

An arithmetic circuit $\Psi$ is said to be \emph{syntactically}
multilinear if for every multiplication gate $v$ in $\Psi$ with children
$u$ and $w$, the sets of variables $X_u$ and $X_w$ are disjoint. We say
that $\Psi$ is \emph{semantically} multilinear if the polynomial computed
at every vertex is a multilinear polynomial. Observe that if $\Psi$
is a syntactically multilinear circuit, then it is also semantically
multilinear. However, it is not clear if every semantically multilinear
circuit can be efficiently simulated by a syntactically multilinear
circuit.

A multilinear circuit is a natural model for computing multilinear
polynomials, but it is not necessarily the most efficient one. Indeed,
it is remarkable that all the constructions of polynomial size arithmetic
circuits for the determinant \cite{Csanky76, Berk84, MV97}, which are
fundamentally different from one another, nevertheless share the property
of being \emph{non}-multilinear, namely, they involve non-multilinear
intermediate computations which eventually cancel out. There are no
subexponential-size multilinear circuits known for the determinant,
and one may very well conjecture these do not exist at all.

Multilinear circuits were first studied by Nisan and
Wigderson~\cite{nw1997}. Subsequently, Raz~\cite{raz2004} defined the
notion of multilinear formulas\footnote{For formulas, it is known that
syntactic multilinearity and semantically multilinearity are equivalent
(See, e.g., \cite{raz2004}).} and showed that any multilinear formula
computing the determinant or the permanent of an $n\times n$ variable
matrix must have super-polynomial size. In a follow up work~\cite{Raz06},
Raz further strengthed the results in~\cite{raz2004} and showed that
there is a family of multilinear polynomials in $n$ variables which can
be computed by a $\poly(n)$ size syntactically multilinear 
arithmetic circuits but
require multilinear formulas of size $n^{\Omega(\log n)}$.

Building on the ideas and techniques developed in~\cite{raz2004}, Raz
and Yehudayoff~\cite{raz-yehudayoff} showed an exponential lower bound
for syntactically multilinear circuits of constant depth. Interestingly,
they also showed a super-polynomial separation between depth $\Delta$ and
depth $\Delta+1$ syntactically multilinear circuits for constant $\Delta$.

In spite of the aforementioned progress on the question of lower bounds
for multilinear formulas and bounded depth syntactically multilinear
circuits, there was no $\Omega(n^{1+\epsilon})$ lower bounds known for
general syntactically multilinear circuits for any constant $\epsilon >
0$. In fact, the results in~\cite{Raz06} show that the main technical
idea underlying the results in~\cite{raz2004, Raz06, raz-yehudayoff}
is unlikely to directly give a super-polynomial lower bound for general
syntactically multilinear circuits. However, a weaker super-linear lower
bound still seemed conceivable via similar techniques.

Raz, Shpilka and Yehudayoff~\cite{RSY08} showed that this is indeed the
case. By a sophisticated and careful application of the techniques
in~\cite{raz2004} along with several additional  ideas, they
established an
$\Omega\left( \frac{n^{4/3}}{\log^2 n}\right)$  lower bound for an
explicit $n$ variate polynomial. Since then, this has remained the
best lower bound known for syntactically multilinear circuits. In this
paper, we improve this result by showing an almost quadratic lower
bound for syntactically multilinear circuits for an explicit $n$ variate
polynomial. In fact, the family of hard polynomials in this paper is the
same as the one used in~\cite{RSY08}. We now formally state our result.

\begin{theorem}
\label{thm:lower-bound-informal}
There is an explicit family of polynomials $\{f_n\}$, where
%$\{f_n : n = 4p \text{ for a
%prime } p\}$ where 
$f_n$ is an $n$ variate multilinear polynomial, such
that any syntactically multilinear arithmetic circuit computing $f_n$
must have size at least $\Omega(n^2/\log^2 n)$.
\end{theorem}

For our proof, we follow the strategy in~\cite{RSY08}. Our improvement
comes from an improvement in a key lemma in~\cite{RSY08} which addresses
the following combinatorial problem.

\begin{question}~\label{ques:unbalancing}
What is the minimal integer $m = m(n)$ for which there is a family of
subsets $S_1, S_2, \ldots, S_m \subseteq [n]$, each $S_i$ satisfying
$6 \log n \leq |S_i| \leq n-6 \log n$ such that for every $T\subseteq
[n], |T| = \floor{n/2}$, there exists an $i \in [m]$ with $|T\cap S_i|
\in \{\floor{|S_i|/2} -3\log n, \floor{|S_i|/2}-3\log n + 1, \ldots,
\floor{|S_i|/2} + 3\log n   \}$?
\end{question}

Raz, Shpilka and Yehudayoff~\cite{RSY08} showed
that $m(n) \geq \Omega\left({n^{1/3}}/{\log n}\right)$. For our proof,
we show that $m(n) \geq \Omega\left({n}/{\log n} \right)$.

In addition to its application to the proof
of~\autoref{thm:lower-bound-informal},~\autoref{ques:unbalancing}
seems to be a natural problem in extremal combinatorics and might be
of independent interest, and special cases thereof were studied in
the combinatorics literature. In the next section, we briefly discuss
the state of the art of this question and state our main technical result
about it in~\autoref{thm:intro:unbalancing}.

\subsection{Unbalancing Sets}

The following question, which is of very similar nature to
\autoref{ques:unbalancing}, is known as Galvin's problem (see \cite{FR87,
EFIN87}): What is the minimal integer $m=m(n)$, for which there exists a
family of subsets $S_1, \ldots, S_m \subseteq [4n]$, each of size $2n$,
such that for every subset $T \subseteq [4n]$ of size $2n$ there exists
some $i \in [m]$ such that $|T \cap S_i| = n$?

It is not hard to show that $m(n) \le 2n$. Indeed, let $S_i = \set{i, i+1,
\ldots, i+2n-1}$, for $i \in \set{1,2,\ldots, 2n+1}$, and let $\alpha_i
(T) = | T \cap S_i | - |([4n]\setminus T) \cap S_i|$. Then $\alpha_i(T)$
is always an even integer, $\alpha_1(T) = -\alpha_{2n+1}(T)$, and
$\alpha_{i} - \alpha_{i+1}(T) \in \set{0, \pm 2}$ if $i \le 2n$. By a
discrete version of the intermediate value theorem, it follows there
exists $j \in [2n]$ such that $\alpha_j(T) = 0$, which implies that
exactly $n$ elements of $S_j$ belong to $T$. Thus, the family $\set{S_1,
\ldots, S_{2n}}$ satisfies this property.

As for lower bounds, a counting argument shows that $m(n) =
\Omega(\sqrt{n})$, since for each fixed $S$ of size $[2n]$ and random $T$
of size $2n$,
\[
\Pr[|T \cap S| = n] = \frac{\binom{2n}{n} \cdot 
\binom{2n}{n}}{\binom{4n}{2n}} = \Theta\left(\frac{1}{\sqrt{n}}\right).
\]

Frankl and R\"odl \cite{FR87} were able to show that $m(n) \ge
\varepsilon n$ for some $\varepsilon >0$ if $n$ is odd, and Enomoto,
Frankl, Ito and Nomura \cite{EFIN87} proved that $m(n) \ge 2n$ if $n$
is odd, which implies that even the constant in the construction given
above is optimal. Until this work, the question was still open for even
values of $n$: in fact, Markert and West (unpublished, see \cite{EFIN87})
showed that for $n \in \set{2,4}$, $m(n) < 2n$.

For our purposes, we need to generalize Galvin's problem in two ways. The
first is to lift the restriction on the set sizes. The second is to
ask how small can the size of the family $\cF = \set{S_1, \ldots, S_m}
\subseteq 2^{[n]}$ be if we merely assume each balanced partition $T$
is ``$\tau$-balanced'' on some $S \in \cF$, namely, if $\abs{|T \cap S|
- |S|/2|} \le \tau$ for some $S$ (the main case of interest for us is
$\tau=O(\log n)$). Of course, since $T$ itself is balanced, very small
or very large sets are always $\tau$-balanced, and thus we impose the
(tight) non-triviality condition $2 \tau  \le |S| \le n-2\tau$ for every $S
\in \cF$.

Once again, by defining $S_i = \set{i, i+1, \ldots, i+n/2-1}$ ($n$
is always assumed to be even), the family $\cF = \set{S_1, S_{1+\tau},
S_{1+2\tau}, ..., S_{1+\lfloor n/(2\tau) \rfloor \cdot \tau} }$ gives a
construction of size $O(n/\tau)$ such that every balanced partition $T$
is $\tau$-balanced on some $S \in \cF$.

It is natural to conjecture
that, perhaps up to a constant, this construction is optimal. Indeed,
this is what we prove here.

\begin{theorem}
\label{thm:intro:unbalancing}
Let $n$ be any large enough even number, and 
let $\tau  \geq 1$ be an integer.
Let $S_1, \ldots, S_m \subseteq [n]$ be sets
such that for all $i \in [m]$, $2 \tau \le |S_i| \le n - 2\tau$. Further,
assume that for every $Y \subseteq [n]$ of size $n/2$ there exists $i
\in [m]$ such that $\abs{|Y \cap S_i| - |S_i|/2} < \tau$. Then, $m \ge
\Omega(n/\tau)$.
\end{theorem}

In particular, \autoref{thm:intro:unbalancing} proves a linear lower
bound $m=\Omega(n)$ for the original problem of Galvin, even when the
universe size is of the form $4k$ for even $k$.

We remark that the relevance of problems of this form to lower bounds
in algebraic complexity was also observed by Jansen \cite{Jansen08}
who considered the problem of obtaining a lower bound on homogenous
syntactically multilinear algebraic branching program (which is a weaker
model than syntactically multilinear circuits), and essentially proposed
\autoref{thm:intro:unbalancing} as a conjecture. In fact, a special case
of this theorem (see~\autoref{thm:unbalancing}), which has a simpler proof,
is already enough to derive the improved lower bounds for syntactically
multilinear circuits.

Alon, Bergmann, Coppersmith and Odlyzko \cite{ABCO88} considered a very
similar problem of balancing $\pm 1$-vectors: they studied families of
vectors $\cF=\set{v_1, \ldots, v_m}$ such that $v_i \in\set{\pm 1}^n$
for $i \in [m]$, which satisfy the properties that for every $w \in
\set{\pm 1}^n$ (not necessarily balanced), there exists $i \in [m]$
such that $|\ip{v_i,w}| \le d$. They generalized a construction of
Knuth \cite{Knuth86} and proved a matching lower bound which together
showed that $m= \lceil n/(d+1) \rceil$ is both necessary and sufficient
for such a set to exist. Galvin's problem seems like ``the $\set{0,1}$
version'' of the same problem, but, to quote from \cite{ABCO88}, there
does not seem to be any simple dependence between the problems.

\subsection{Proof overview}
In this section, we discuss the main ideas and give a brief
sketch of the proofs of~\autoref{thm:lower-bound-informal}
and~\autoref{thm:intro:unbalancing}. Since our proof heavily depends
on the proof in~\cite{RSY08} and follows the same strategy, we start by
revisiting the main steps in their proof and noting the key differences
between the proof in~\cite{RSY08} and our proof. We also outline the
reduction to the combinatorial problem of unbalancing set families
in~\autoref{ques:unbalancing}.

\subsubsection*{Proof sketch of~\cite{RSY08}}
The proof in~\cite{RSY08} starts by proving a syntactically multilinear
analog of a classical result of Baur and Strassen~\cite{BS83}, where
it was shown that if an $n$ variate polynomial $f$ is computable by an
arithmetic circuit $\Psi$ of size $s(n)$, then there is an arithmetic
circuit $\Psi'$ of size at most $5s(n)$  with $n$ outputs such that
the $i$-th output gate of $\Psi'$ computes $f_i = \frac{\partial
f}{\partial x_i}$. Raz, Shpilka and Yehudayoff show that if $\Psi$
is syntactically multilinear, then the circuit $\Psi'$ continues to
be syntactically multilinear. Additionally, there is no directed
path from a leaf labeled by $x_i$ to the output gate computing
$f_i$.\footnote{See~\autoref{thm:BS-multilinear} for a formal statement.}

Once we have this structural result, it would suffice to prove a lower
bound on the size of $\Psi'$. For brevity, we denote the subcircuit of
$\Psi'$ rooted at the output gate computing $f_i$ by $\Psi_i'$. As a
key step of the proof in~\cite{RSY08}, the authors identify certain sets
of vertices $\cU_1,\cU_2, \ldots, \cU_n$ in $\Psi'$ with the following
properties.
\begin{itemize}
\item For every $i\in [n]$,  $\cU_i$ is a subset of  vertices in  $\Psi_i'$.
\item For every $i \in [n]$ and $v\in \cU_i$, the number of $j \neq i$ such that $v \in \cU_j$ is not too large (at most $O(\log n)$). 
\end{itemize}
Observe that at this point, showing a lower bound of $s'(n)$ on the
size of each $\cU_i$ implies a lower bound of 
$\Omega(ns'(n)/{\log n})$ on the
size of $\Psi'$ and hence $\Psi$. In~\cite{RSY08}, the authors show
that there is an explicit $f$ such that each $\cU_i$ must have size
at least $\Omega(n^{1/3}/\log n)$, thereby getting a lower bound of
$\Omega(n^{4/3}/\log^2 n)$ on the size of $\Psi$.

For our proof, we  follow precisely this high level strategy. Our
improvement in the lower bound comes from showing that each $\cU_i$ must
be of size at least $\Omega(n/\log n)$ and not just $\Omega(n^{1/3}/\log
n)$ as shown in~\cite{RSY08}. We now elaborate further on the main
ideas in this step in~\cite{RSY08} and the differences with the proofs
in this paper.

We start with some intuition into the  definition of the sets $\cU_i$
in~\cite{RSY08}. Consider a vertex $v$ in $\Psi'$ which depends on at
least $k$ variables. Without loss of generality, let these variables
be $\{x_1, x_2, \ldots, x_k\}$. From~\autoref{item:key-property}
in~\autoref{thm:BS-multilinear}, we know that the variable $x_i$ does
not appear in the subcircuit $\Psi_i'$. Therefore, the vertex $v$ cannot
appear in the subcircuits $\Psi_1', \Psi_2', \ldots, \Psi_k'$. So, if we
define the set $\cU_i$ as the set of vertices in $\Psi_i'$ which depend on
at least $k$ variables, then $\cU_i$ must be disjoint from vertices in at
least $k$ of the subcircuits $\Psi_1', \Psi_2', \ldots, \Psi_n'$. Picking
$k \geq n-O(\log n)$ would give us the desired property. So, if we
can prove a lower bound on the size of the set $\cU_i$, we would be
done. However, the definition of the set $\cU_i$ so far turns out to be
too general, and we do not know a way of directly proving a lower bound
on its size.\footnote{Indeed, it is not even immediately clear if the
$\cU_i$ has any other gates apart from the output gate of $\Psi_i'$.}

To circumvent this obstacle,~\cite{RSY08} define the set $\cU_i$ (called
the \emph{upper leveled} gates in $\Psi_i'$) as the set of all vertices in
$\Psi_i'$ which depend on at least $n-6\log n$ variables and have a child
which depends on more than $6\log n$ variables and less than $n-6\log n$
variables. This additional structure is helpful in proving a lower bound
on the size of $\cU_i$. We now discuss this in some more detail.

For every $i \in [n]$, let $\cL_i$ be the set of vertices $u$ in
$\Psi_i'$, such that $6 \log n < |X_u| < n-6 \log n$, and $u$ has a
parent in $\cU_i$. These gates are referred to as \emph{lower leveled}
gates. Observe that $|\cU_i| \geq \frac{|\cL_i|}{2}$, since the in-degree
of every vertex in $\psi_i'$ is at most $2$. The key structural property
of the set $\cL_i$ is the following (see Proposition 5.5 in~\cite{RSY08}).

\begin{lemma}[\cite{RSY08}]~\label{lem:lower-level-decomposition}
Let $i \in [n]$, and let $h_1, h_2, \ldots, h_{\ell}$ be the polynomials computed by the gates in $\cL_i$. Then, there exist multilinear polynomials $g_1, g_2, \ldots, g_{\ell}, g$ such that   
\begin{equation}\label{eqn:lower-level-decomposition}
f_i = \sum_{j \in [\ell]} g_j\cdot h_j  + g 
\end{equation}
where 
\begin{itemize}
\item For every $j \in [\ell]$, $h_j$ and $g_j$ are variable disjoint.
\item The degree of $g$ is at most $O(\log n)$.
\end{itemize}
\end{lemma}

Observe that~\autoref{eqn:lower-level-decomposition} is basically
a decomposition of a potentially-hard polynomial $f_i$ in terms of
the sum of products of  multilinear polynomials in an intermediate
number of variables. The goal is to show that for an appropriate
explicit $f_i$, the number of summands on the right hand side
of~\autoref{eqn:lower-level-decomposition} cannot be too small. A
similar scenario also appears in the multilinear formula lower bounds
and bounded depth multilinear formula lower bounds of \cite{raz2004,
Raz06, raz-yehudayoff} (albeit with some key differences). Hence, a
natural approach at this point would be to use the tools in~\cite{raz2004,
Raz06, raz-yehudayoff}, namely  the rank of the \emph{partial derivative
matrix}, to attempt to prove this lower bound. We refer the reader
to~\autoref{sec:partial-derivative-matrix} for the definitions
and properties of the partial derivative matrix and proceed with
the overview. For each $j \in [\ell]$, let the polynomial $h_j$
in~\autoref{lem:lower-level-decomposition} depend on the variables $S_j
\subseteq X$. The key technical step in the rest of the proof is to
show that there is a partition of the set of variables $X = \{x_1, x_2,
\ldots, x_n\}$ into $Y$ and $Z$ such that $|Y| = |Z|$ and for every
$j \in [\ell]$, $\abs{|S_j \cap Y| - |S_j \cap Z|} \geq \Omega(\log
n)$. In~\cite{RSY08}, the authors show that there is an absolute constant
$\epsilon > 0$ such that if $\ell \leq \epsilon n^{1/3}/\log n$, then
there is an equipartition of $X$ which \emph{unbalances} all the sets
$\lbrace S_j : j \in [\ell] \rbrace$ by at least $\Omega(\log n)$. Our key
technical contribution (\autoref{thm:intro:unbalancing}) in this paper
is to show that as long as $\ell \leq \epsilon n/\log n$, there is an
equipartition which unbalances all the $S_j$'s by at least $\Omega(\log
n)$. This implies an $\Omega(n / \log n)$ on the size of each set $\cU_i$,
and thus an $\Omega(n^2/\log^2 n)$ lower bound on the circuit size.

Before we dive into a more detailed discussion on the overview and
main ideas in the proof of~\autoref{thm:intro:unbalancing} in the
next section, we would like to remark that  the lower bound question
in~\autoref{eqn:lower-level-decomposition} seems to be a trickier
question than what is encountered  while proving multilinear formula lower
bounds~\cite{raz2004, Raz06} or bounded depth syntactically multilinear
circuit lower bounds~\cite{raz-yehudayoff}. The main differences are
that in the proofs in~\cite{raz2004, Raz06, raz-yehudayoff}, the sets
$S_j$ have a stronger guarantee on their size (at least $n^{\Omega(1)}$
and at most $n-n^{\Omega(1)}$), and each of the summands on the right
has \emph{many}  variable disjoint factors and not just two factors as
in~\autoref{eqn:lower-level-decomposition}. For instance, in the formula
lower bound proofs the number of variable disjoint factors in each summand
on the right is $\Omega(\log n)$, and  for constant depth circuit lower
bounds it is $n^{\Omega(1)}$. Together, these properties make it possible
to show much stronger lower bounds on $\ell$. In particular, it is known
that a \emph{random} equipartition works for these two applications, in
the sense that it unbalances sufficiently many factors in each summand,
thereby implying that the rank of the partial derivative matrix of the
polynomial is small. Hence,  for an appropriate\footnote{$f_i$ is chosen
so that the the partial derivative matrix for $f_i$ is of full rank
for \emph{every} equipartition.} $f_i$, the number of summands must
be large. However, since a set of size $O(\log n)$ is balanced under
a random equipartition with probability $\Omega(1/\sqrt{\log n})$ and
the identity in~\autoref{eqn:lower-level-decomposition} involves just
two variable disjoint factors, taking a random equipartition would not
enable us to prove any meaningful bounds.

\iffalse
The two main reasons for this are as follows. 
\begin{itemize}
\item In the proofs in~\cite{raz2004, Raz06, raz-yehudayoff} 
the sets of variables we want to unbalance using an equipartition have 
\end{itemize} 

Here we could say something along these lines: after Baur-Strassen,
the challenge is to count number of gates in each of the $n$ subcircuits
without overcounting. Gates with very high support are a natural candidate
because they are not connected to many outputs. However, because they
have high support, they are ``too strong'' to be handled atomically. (to
explain this we probably need to discuss partial derivative matrix). this
leads to the definition of ``upper-leveled'' gates (which we want to
count) and ``lower-leveled gates'' (which are weak enough so that we
can actually count). If you think this is too much at this point, this
discussion can be moved to section 4.

Another important thing to mention is that the formula and bounded depth
lower bounds proofs lead to a different unbalancing problem. there is
it possible to show that a random partition works. but in our case,
since the sets can be pretty small or pretty large, a random partition
would not work.
\fi

\subsubsection*{Proof sketch of~\autoref{thm:intro:unbalancing}}

Recall that our task is, given a small collection of subsets of
$[n]$, to find a balanced partition which is unbalanced on each of the
sets. Equivalently, we would like to prove that 
if $\cF$ is a family of subsets
such that every balanced partition balances at least one set in $\cF$,
then $|\cF|$ must be large (of course, $\cF$ must satisfy the conditions
in \autoref{thm:intro:unbalancing}).

We first sketch the proof of a special case (which suffices for the
main application here), when $n=4p$ and $p$ is a prime. 
For the sake of simplicity, suppose also that all subsets $S \in \cF$ are of
even size, and assume further that for every subset $T \subseteq [n]$
of size $n/2$ there exists $S \in \cF$ such that $T$ completely balances
$S$, namely, $|T \cap S| = |S|/2$. One possible approach to obtain lower
bounds on $|\cF|$ is via an application of the polynomial method as
done, for example, in \cite{ABCO88}. Define
the following polynomial over, say, the rationals:
\[
f(x_1, \ldots, x_n) = \prod_{S \in \cF} (\ip{x,\ind_S} - |S|/2).
\]
By the assumption on $\cF$, the polynomial $f$ evaluates to $0$ over
all points in $\set{0,1}^n$ with Hamming weight exactly $n/2$. We can
also argue, using the assumption on the set sizes in $\cF$, that $f$
is not identically zero, and clearly $\deg(f) \le |\cF|$. Thus, a lower
bound on $\deg(f)$ translates to a lower bound on $|\cF|$.

This idea, however, seems like a complete nonstarter, since there exists
a degree $1$ non-zero polynomial which evaluates to 0 over the middle
layer of $\set{0,1}^n$, namely, $\sum_i x_i - n/2$.

A very clever solution to this potential obstacle was found by Heged\H{u}s
\cite{Heg10}. Suppose $n=4p$ for some prime $p$. The main insight in
\cite{Heg10} is to consider the polynomial $f$  over $\F_p$, and
to add the requirement that there exists some $z \in \set{0,1}^{4p}$,
of Hamming weight \emph{exactly} $3p$, such that $f(z) \neq 0$. This
requirement rules out the trivial example $\sum_i x_i - n/2$, and
Heged\H{u}s was able to show that the degree of any polynomial with
these properties must be at least $p=n/4$ (see \autoref{lem:hegedus}
for the complete statement).

We are thus left with the task of proving that our polynomial evaluates
to a non-zero value over some point $z \in \set{0,1}^{4p}$ of Hamming
weight $3p$. This turns out to be not very hard to show, assuming each
set is of size at least, say, $100 \log n$ and at most $n - 100 \log n$, by
choosing a random such vector $z$. Indeed, it is not surprising that it
is much easier to directly show that a highly unbalanced partition of
$[n]$ (into $3n/4$ vs $n/4$) unbalances all the sets $\cF$.\footnote{In
our case, we need to argue that the imbalance is non-zero modulo $p$,
which adds an extra layer of complication, although again, one which is
not hard to solve.}

As mentioned earlier, the case $n=4p$ and $\tau \ge 100 \log n$ in
\autoref{thm:intro:unbalancing} is considerably easier to prove and
suffices for the application to circuit lower bounds. Proving this
theorem for every even $n$ and every $\tau \ge 1$ requires further
technical ideas. We postpone this discussion to \autoref{sec:general}.

Even though~\autoref{lem:hegedus} seems to be a fundamental statement
about polynomials over finite fields and could conceivably have
an elementary proof, the proof in~\cite{Heg10} uses more advanced
techniques. It relies on the description of Gr\"obner basis for ideals of
polynomials in $\F[x_1, x_2, \ldots, x_{n}]$ which vanish on all points
in $\{0,1\}^n$ of weight equal to $n/2$. A complete description of the
reduced Gr\"obner basis for such ideals was given by Heged\H us and
R\'onyai~\cite{HR03} and their proof builds up on a number of earlier
partial results~\cite{ARS02, FG06} on this problem.

To the best of our knowledge, the proof in~\cite{Heg10} is the only known
proof of~\autoref{lem:hegedus}, and giving a self contained elementary
proof of it seems to be an interesting question.

\iffalse
This would be a good place to mention the Gr\"obner bases techniques of \cite{Heg10}, \cite{HR03}, \cite{ARS02}. Probably also to mention that giving a ``sane'' proof for Heged\H{u}s' lemma is an interesting problem.
\fi

\subsection*{Organization of the paper}
In the rest of the paper, we set up some notation and
discuss some preliminary notions in~\autoref{sec:prelim},
prove~\autoref{thm:intro:unbalancing} in~\autoref{sec:unbalancing}
and complete the proof of~\autoref{thm:lower-bound-informal}
in~\autoref{sec:proof of main theorem}. Throughout the paper we
assume, whenever this is needed, that $n$ is sufficiently large, and
make no attempts to optimize the absolute constants.

\section{Preliminaries}\label{sec:prelim}
For $n \in \N$, we denote $[n]=\set{1,2,\ldots, n}$. For a prime $p$,
we denote by $\F_p$ the finite field with $p$ elements. For
two integers $i, j$ with $i\leq j$, we denote $[i, j] = \set{a \in \Z :
i \leq a \leq j}$.
%Also, we use $(i,j)$ for the set $(i,j) = \set{a\in \Z : i < a < j}$.
The characteristic vector of a set $S \subseteq [n]$
is denoted by $\ind_S \in \set{0,1}^n$.

As is standard, $\binom{[n]}{k}$ denotes the family $\setdef{S \subseteq
[n]}{|S|=k}$.

For an even $n \in \N$ and $Y \subseteq [n]$ such that $|Y|=n/2$,
we call $Y$ a \emph{balanced partition} of $[n]$, with the implied
meaning that $Y$ partitions $[n]$ evenly into $Y$ and $[n]\setminus Y$.
The \emph{imbalance} of a set $S \subseteq [n]$ under $Y$ is $d_Y (S) :=
\abs{ |Y \cap S| - |S|/2}$. Observe the useful symmetry $d_Y(S) =d_Y([n]
\setminus [S])$, which follows from the fact that $|Y|=n/2$. We say $S$
is $\tau$-unbalanced under $Y$ if $d_Y(S) \ge \tau$.

We use the following lemma from \cite{Heg10}.

\begin{lemma}[\cite{Heg10}]
\label{lem:hegedus}
Let $p$ be a prime, and let $f \in \F_p[x_1, \ldots, x_{4p}]$ be a polynomial. Suppose that for all $Y \in \binom{[4p]}{2p}$, it holds that $f(\ind_Y) = 0$, and that there exists $T \subseteq [4p]$ such that $|T|=3p$ and $f(\ind_T) \neq 0$. Then $\deg(f) \ge p$.
\end{lemma}

\subsection{Hypergeometric distribution}
For parameters $N, M, k$, where $N \geq M$,  by ${\cal H}(M, N, k)$, we denote the distribution of $\abs{S\cap T}$, where $S$ is any fixed subset of $[N]$ of size $M$, and $T$ is a uniformly random subset of $[N]$ of size equal to $k$. Clearly, 
\[
\Pr[\abs{S \cap T} = i] = \frac{\binom{M}{i} \binom{N-M}{k-i}}{\binom{N}{k}} \, .
\]
The expected value of $\abs{S \cap T}$ under this distribution is equal to $kM/N$. We need the following tail bound of hypergeometric distribution for our proof.  
\begin{lemma}[\cite{Skala13}] \label{lem : hypergeometric tail}
Let $N, M, k$, and ${\cal H}(M, N, k)$ be as defined above. Then, for every $t$
\[
\Pr[\abs{\abs{S \cap T} - kM/N} \geq tk] \leq \me^{-2t^2 k} \, .
\]
\end{lemma}

\begin{lemma}[Hoeffding's inequality, \cite{AlonSpencer}]\label{lem : hoeffding 1}
Let $X_1, X_2, \ldots, X_n$ be independent random variables taking values in $\{0,1\}$. Then, 
\[
\Pr\left[\abs{\sum_{i =1}^n X_i - \E[\sum_{i = 1}^n X_i]} \geq t\right] \leq 2\exp(-2t^2/n) \, .
\]
\end{lemma}

\subsection{Partial derivative matrix}\label{sec:partial-derivative-matrix}
For a circuit $\Psi$, we denote by $|\Psi|$ the size of $\Psi$, namely, the number of gates in it. For a gate $v$, we denote by $X_v$ the set of variables that occur in the subcircuit rooted at $v$.

Let $X=\set{x_1, \ldots, x_n}$ be a set of variables, $Y \subseteq X$ (not necessarily of size $n/2$) and let $Z=X \setminus Y$. For a multilinear polynomial $f(X) \in \F[X]$, we define the \emph{partial derivative matrix} of $f$ with respect to $Y,Z$, denoted $M_{Y,Z}(f)$, as follows: the rows of $M$ are indexed by multilinear monomials in $Y$. the columns of $M$ are indexed by multilinear monomials in $Z$. The entry which corresponds to $(m_1, m_2)$ is the coefficient of the monomial $m_1 \cdot m_2$ in $f$. We define $\rank_{Y,Z}(f) = \rank(M_{Y,Z}(f))$.

The following properties of the partial derivative matrix are easy to prove and well-documented (see, e.g., \cite{RSY08}).

\begin{proposition}
\label{prop:pdmatrix}
The following properties hold:
\begin{enumerate}
\item \label{item:trivial-upper-bound} For every multilinear polynomial $f(X) \in \F[X]$, $Y \subseteq X$ and $Z=X \setminus Y$, $\rank_{Y,Z}(f) \le \min\set{2^{|Y|}, 2^{|Z|}}$.
\item \label{item:rank-additive} For every two multilinear polynomials $f_1(X), f_2(X) \in \F[X]$ and for every partition $X=Y \sqcup Z$, $\rank_{Y,Z}(f_1 + f_2) \le \rank_{Y,Z}(f_1) + \rank_{Y,Z}(f_2)$.
\item \label{item:rank-multiplicative}
Let $f_1 \in \F[X_1]$ and $f_2 \in \F[X_2]$ be multilinear polynomials such that $X_1 \cap X_2 = \emptyset$. Let $Y_i \subseteq X_i$ and $Z_i = X_i \setminus Y_i$ for $i \in \set{1,2}$. Set $Y = Y_1 \cup Y_2, Z=Z_1 \cup Z_2$. Then $\rank_{Y,Z}(f_1 \cdot f_2) = \rank_{Y_1,Z_1} (f_1) \cdot \rank_{Y_2,Z_2} (f_2)$.
\item \label{item:full-rank-derivative} Let $f(X) \in \F[X]$ be a multilinear polynomial such that $X=Y \sqcup Z$ and $|Y|=|Z|=n/2$. Suppose $\rank_{Y,Z}(f) = 2^{n/2}$, and let $g = \partial f / \partial x$ for some $x \in X$. Then $\rank_{Y,Z}(g) = 2^{n/2 - 1}$.
\item \label{item:low-degree-rank} Let $f(X) \in \F[X]$ be a multilinear polynomial of total degree $d$. Then for every partition $X=Y \sqcup Z$ such that $|Y|=|Z|=n/2$, $\rank_{Y,Z}(f) \le 2^{(d+1) \log (n/2)}$.
\end{enumerate}

\end{proposition}

\section{Unbalancing sets under a balanced partition}\label{sec:unbalancing}
In this section, we prove~\autoref{thm:intro:unbalancing}. We start by
proving a special case (see~\autoref{thm:unbalancing} below) when $n$
equals $4p$ for some prime $p$, and $\tau \geq \Omega(\log n)$. This
special case already suffices for the application to the proof
of~\autoref{thm:lower-bound-informal} (for infinitely many values of
$n$), and has a somewhat simpler
proof. We then move on to prove the case for general $n$ and $\tau$,
which while being similar to the proof of~\autoref{thm:unbalancing},
needs some additional ideas and care.
\subsection{Special case : $n = 4p$ and $\tau \geq \Omega(\log n)$}
\begin{theorem}
\label{thm:unbalancing}
Let $p$ be a large enough prime, and let $\log p \le \tau \le p/1000$. Let
$S_1, \ldots, S_m \subseteq [4p]$ be sets such that for all $i \in [m]$,
$100 \tau \le |S_i| \le 4p - 100\tau$. Further, assume that for every
balanced partition $Y$ of $[4p]$ there exists $i \in [m]$ such that
$d_Y(S_i) < \tau$. Then, $m \ge \frac{1}{2} \cdot p/\tau$.
\end{theorem}

We start with the following lemma, which shows that a small collection
of sets can be unbalanced (modulo $p$) by a partition which is very
unbalanced.

\begin{lemma}
\label{lem:3p}
Let $p$ be a large enough prime, and let $\log p \le \tau \le p/1000$. Let
$S_1, \ldots, S_m \subseteq [4p]$ be sets such that for all $i \in [m]$,
$100 \tau \le |S_i| \le 2p$. Assume further $m \le p$. Then, there exists
$T \subseteq [4p]$, $|T|=3p$ such that for all $i \in [m]$ and for all
$-\tau +1 \le t \le \tau $, $|S_i \cap T| \not\equiv \floor{|S_i|/2} +
t \bmod p$.
\end{lemma}

To prove \autoref{lem:3p}, we use the following two technical claims. Let
$\mu_{3/4}$ denote the probability distribution on subsets of $[4p]$
obtained by putting each $j \in [4p]$ in $T$ with probability $3/4$,
independently of all other elements.

\begin{claim}
\label{cl:fine-tuned}
For a random set $T \sim \mu_{3/4}$, $\Pr[|T|=3p] = \Theta(1/\sqrt{p})$.
\end{claim}

\begin{proof}
The probability that $|T|=3p$ is given by $\binom{4p}{3p} \cdot
(3/4)^{3p} \cdot (1/4)^p$, which is $\Theta(1/\sqrt{p})$, by Stirling's
approximation.
\end{proof}

\begin{claim}
\label{cl:unbalances-set}
Let $\log p \le \tau \le p/1000$ and let $S \subseteq [4p]$ such that
$100 \tau \le |S| \le 2p$. For a random set $T \sim \mu_{3/4}$, the
probability that for some integer $-\tau + 1 \le t \le \tau$ it holds
that $|T \cap S_i| = \floor{|S_i|/2} + t \bmod p$ is at most $1/p^5$.
\end{claim}

\begin{proof}
Denote $s = |S|$. Then $\E[|T \cap S|] = 3s/4$. We say $T$ is bad for $S$
if $|T \cap S| = \floor{s/2} + t + kp$ for some $-\tau  \le t \le \tau +
1$ and $k \in \Z$. We claim this in particular implies that $\abs{ |T
\cap S_i| - 3s/4} \ge s/5$. Indeed, since $|T \cap S|$ is an integer in
the interval $[0,2p]$, and by the bounds on $s$, the only cases needed
to be analyzed are $k=0, \pm 1$.

If $|T \cap S| = \floor{s/2} +t - p$, 
then clearly $|T \cap S| \le \floor{s/2}$ which implies the statement.

If $|T \cap S| = \floor{s/2} +t + p$, 
then, as $s \le 2p$ and $\tau \le s/100$,
\[
|T \cap S| - 3s / 4  \ge -s/4 -1 + t + p 
\ge p/2 + t - 1 \ge s/4 + t - 1 \ge s/5
\]
(The ``$-1$'' accounts for the fact that $s/2$ might not be an integer).

Finally, if $|T \cap S| = \floor{s/2} +t$, it holds that \[
|T \cap S| \le s/2 + \tau \le s/2 + 2s/100,
\]
which again implies the statement.

By Chernoff Bound (see, e.g., \cite{AlonSpencer}), $\Pr[\abs{ |T \cap S_i|
- 3s/4} \ge s/5] \le 2^{-|S|/20} \le 1/p^5$, hence $T$ is bad for $S$
with at most that probability.
\end{proof}

The proof of \autoref{lem:3p} is now fairly immediate.

\begin{proof}[Proof of \autoref{lem:3p}]
Pick $T \sim \mu_{3/4}$.  By \autoref{cl:fine-tuned}, $|T|=3p$ with
probability $\Theta(1/\sqrt{p})$.  Recall that $T$ is bad for $S_i$ if $|T
\cap S_i| = \floor{|S_i|/2} + t \bmod p$ for $t \in \set{-\tau+1, \ldots, \tau}$.  By \autoref{cl:fine-tuned},
for each $S_i$, $T$ is bad for $S_i$ with probability at most $1/p^5$.
Hence, the probability that there exists $i \in [m]$ such that $T$
is bad for $S_i$ is at most $m /p^5 \le 1/p^4$.

It follows that with probability at most $1-\Theta(1/\sqrt{p}) + 1/p^4 <
1$, either $|T| \neq 3p$ or $T$ is bad for some $S_i$, and hence there
exists a selection of $T$ such that $|T|=3p$ and $T$ is good for all
$S_i$'s.
\end{proof}

We are now ready to prove \autoref{thm:unbalancing}.

\begin{proof}[Proof of \autoref{thm:unbalancing}]
Let $S_1, \ldots, S_m$ be a collection of sets as stated in the
theorem. Since $d_Y(S_j) = d_Y([n]\setminus S_j)$, we can assume without
loss of generality, by possibly replacing a set with its complement,
that $|S_j| \le 2p$ for all $j \in [m]$. We may further assume $m \le p$
as otherwise the statement directly follows.  For $j \in [m]$, define
the following polynomials over $\F_p$:
\[
B_j (x_1, \ldots, x_{4p}) 
= \prod_{t=-\tau+1}^{\tau} (\ip{x,\ind_{S_j}} - \floor{|S_j|/2} - t),
\]
where $x = (x_1, \ldots, x_{4p})$ 
and $\ip{u,v} = \sum u_i v_i$ is the usual inner product. Further, define
\[
f(x_1, \ldots, x_{4p}) = \prod_{j=1}^m B_j (x_1, \ldots, x_{4p}),
\]
as a polynomial over $\F_p$.

By assumption, for every $Y \in \binom{[4p]}{2p}$, $f(\ind_Y) =
0$. This follows because $\ip{\ind_Y, \ind_{S_j}} = |Y \cap S_j|$,
and by assumption, for some $j$ is holds that $d_Y(S_j) < \tau$, so it
must be that $|Y \cap S_j| - \floor{|S_j|/2} \in \set{-\tau+1, \ldots, 0,
\ldots, \tau}$, so that $B_j(\ind_Y) = 0$.

Furthermore, \autoref{lem:3p} guarantees the existence of a set $T \in
\binom{[4p]}{3p}$ such that $f(\ind_T) \neq 0$, as the set $T$ from
\autoref{lem:3p} satisfies the property that $(\ip{\ind_T,\ind_{S_j}} -
\floor{|S_j|/2} - t) \neq 0 \bmod p$ for all $-\tau + 1 \le t \le \tau $
and for all $j \in [m]$.

By \autoref{lem:hegedus}, $\deg(f) \ge p$, and by construction, $\deg(f)
\le  2 \tau \cdot m$, which implies the desired lower bound on $m$.
\end{proof}

\subsection{General $n$ and $\tau$}\label{sec:general}

In this section, we extend \autoref{thm:unbalancing} for a more general
range of parameters, by proving the following.

\begin{theorem}\label{thm:unbalancing-general}
Let $n$ be a large enough even natural number, and let $\tau \in \{1, 2,
\ldots, n/10^6\}$ be a parameter. Let $S_1, S_2, \ldots, S_m \subseteq
[n]$ be sets such that for each $i \in [m]$, $2\tau \leq |S_i| \leq
n-2\tau$. Furthermore, assume that for every balanced partition $Y$ of
$[n]$, there exists an $i$ such that $d_Y(S_i) < \tau$. Then, $m \geq
\frac{1}{10^5}\cdot n/\tau$.
\end{theorem} 

We remark that \autoref{thm:unbalancing} suffices for the application
to circuit lower bounds, and thus, a reader who is more interested in
that aspect of this work may safely skip to \autoref{sec:proof of main
theorem}.

Recall that in \autoref{thm:unbalancing} we have required the universe
size $n$ to be of the form $4p$ for a prime $p$, and the sets $S_1,
\ldots, S_m$ to be of size at least logarithmic in $n$ (as commented
earlier, we may assume $|S_i| \le n/2$ for every $i$, by possibly
replacing $S_i$ with its complement).

Our strategy for general even\footnote{In order to talk about balanced
partitions of the universe, $n$ clearly must be even. However, our
techniques can be easily extended to odd integers, if one is willing
to replace balanced partitions by almost-balanced partitions, that is,
partitions $[n]=Y \sqcup Z$ such that $|\abs{Y}-\abs{Z}| = 1$. We omit
the straightforward details.} $n$ and general $\tau$ will be very
similar for the previous special case. In order to apply the useful
\autoref{lem:hegedus}, we start by ``forcing'' the universe size to be
of the form $4p$. This is done by picking the largest number of the form
$4p$ which is smaller than $n$ (known results about the distribution
of prime numbers guarantee the existence of such a prime such that $n-4p
\le n^{0.6}$). We then randomly pick a subset of $A \subset [n]$ of size
$n-4p$ avoiding all the small sets 
and partition $A$ in an arbitrary balanced manner. Such a subset
is guaranteed, with high probability, to have a small intersection
with every $S_i$, and thus for every such set the values of very few
elements have been determined.  Again, this intersection property is
easier to show, by standard concentration bounds, when the sets $S_i$
are somewhat large, whereas in our case they can be small. However, the
fact that $|A|$ itself is sublinear in $n$ enables us to handle all cases.

We now denote $\tilde{S}_i = S_i \setminus A$ and $\widetilde{[n]} = [n]
\setminus A$, and, as before, we would like to find a set $T \subseteq
\widetilde{[n]}$ of size exactly $3p$ that is unbalanced, modulo $p$,
on every $\tilde{S}_i$ (and since $\tilde{S}_i$ is a very large subset
of $S_i$, this property will extend to $S_i$ itself). A na\"{i}ve random
choice, as is done in the proof of \autoref{thm:unbalancing}, will not
work, since the probability of failure for very small sets will be too
large to apply a union bound over all sets. Thus, we pick $T$ using a
different, and slightly more complicated, random procedure.

Given such $T$ and $A$, the proof follows from a similar construction
of a polynomial in a similar application of \autoref{lem:hegedus}. We
now provide the details.

We start by proving the existence of a set $A$ as described above.

\begin{lemma}\label{lem : reducing to 4p}
Let $\tau \geq 1$ be an integer and $S_1, S_2, \ldots, S_m$ be subsets of
$[n]$, such that $m \leq 10^{-5}n/\tau$. Then, for every integer $a \le
n^{0.6}$, there exists an $A \subseteq [n]$ of size exactly $a$ such that
for every $i \in [m]$, $\abs{A\cap S_i} \leq  0.01\abs{S_i}$. Moreover,
for each $i \in [m]$, if $\abs{S_i} \leq 10^4\tau$, then $A \cap S_i
=\emptyset$.
\end{lemma}

%\begin{lemma}\label{lem : reducing to 4p}
%Let  $S_1, S_2, \ldots, S_m$ be subsets of $[n]$, such that $m \leq n/10^5$. Then, for every integer $a \le n^{0.6}$, there exists an $A \subseteq [n]$ of size exactly $a$ such that for every $i \in [m]$, $\abs{A\cap S_i} \leq 0.01\abs{S_i}$. 
%\end{lemma}
\begin{proof}
Let $L = \bigcup_{i : \abs{S_i} \leq 10^4\tau} S_i$ and let $\ell =
\abs{L}$. Since $m \leq 10^{-5}n$, we know that  $ \ell \leq m \cdot 10^4
\leq n/10$. Let $A$ to be a uniformly random subset of $[n]\setminus L$
of size $a$.

We now show that with high probability $A$ satisfies $\abs{A\cap S_i}\leq
0.01\abs{S_i}$ for every $i \in [m]$.  We consider three cases.

\begin{itemize}
\item {\bf Small sets: $ \abs{S_j} \leq 10^4\tau $. } By the choice of $A$, we know that $A$ is disjoint from all subsets of size at most $10^4\tau$. 

\item {\bf Large sets: $\abs{S_j} \geq n^{0.31}$. } 
For any fixed set $S_i$ of size at least $n^{0.31}$, by~\autoref{lem : hypergeometric tail}, we know that
\[
\Pr\left[\abs{A \cap S_i} - \abs{A}\abs{S_i}/(0.9n) \geq 0.009\abs{S_i} \right] \leq \exp(-\Omega(\abs{S_i}^2/|A|)) \, .
\]
Since $\abs{A} \leq n^{0.6}$ and $\abs{S_i} \geq n^{0.31}$, this
probability is at most $\exp(-\Omega(n^{0.02}))$. Thus, by a union bound,
we know that with probability at least $1-\exp(-\Omega(n^{0.02}))$,
for each $S_i$ with $\abs{S_i} \geq n^{0.31}$, $\abs{A \cap S_i} \leq
0.01\abs{S_i}$.

\item {\bf Sets of intermediate size: $10^4\tau  \le \abs{S_j} \leq
n^{0.31}$. }  We now argue that for all such sets, $\abs{A\cap S_i}
\leq 100$, with high probability.

To this end, we first upper bound the probability that the set $A$
contains a fixed set $S$ of size $100$, and then take a union bound
over all sets $S$ of size $s = 100$ which are a subset of some $S_i$
of intermediate size. Let $S$ be a fixed set of size $100$. Then,

\begin{align*}
\Pr[S\subseteq A] &\leq \frac{\binom{n-\ell - s}{a-s}}{\binom{n-\ell}{a}}  \\
&=  \frac{(n-\ell-s)!}{(a-s)!(n-\ell - a)!} \cdot  \frac{a! (n-\ell - a)!}{(n-\ell)!}  \\
&=  \frac{(n-\ell-s)!}{(a-s)!} \cdot  \frac{a! }{(n-\ell)!}  \\
&=  \frac{(n-\ell-s)!}{(n-\ell)!} \cdot  \frac{a! }{(a-s)!}  \\
& \leq  \left( \frac{a}{n-\ell-s} \right)^s   \\
& \leq \left( \frac{n^{0.6}}{n-0.1n - n^{0.6}}\right)^{s} \quad \quad (\text{using bounds on } \ell \text{ and } a)  \\
&\leq   n^{-0.39 s}  \\
&\leq   n^{-39 } \quad \quad (\text{using s = 100})
%& \leq & \left( \frac{1}{0.8 n^{0.4}} \right)^{100 \tau} \, .
\end{align*}
For each $S_i$ of size at most $n^{0.31}$ there are at most $(n^{0.31})^{100}$ subsets of size $100$.
Therefore, by a union bound, the probability that $\abs{A\cap S_i}
\geq
100$ for any subset $S_i$ of size at most $n^{0.31}$ is at most $n^{-39}
\cdot n \cdot n^{31}=n^{-7}$.
\end{itemize}
A union bound over all three cases completes the proof of the lemma.
\end{proof}

Having shown the existence of the set $A$ as described in the proof
outline, we turn to show the existence of a set $T$.

\begin{lemma}\label{lem:3p general}
Let $n$ be a natural number, $p$ be a prime satisfying $n - n^{0.6}
\leq 4p \leq n $ and let $\tau$ be an integer satisfying $1 \leq \tau \leq
p/10^5$. Let $S_1, S_2, \ldots, S_m$ be subsets of $[n]$, such that $m
\leq 10^{-5}n/\tau$ and for every $j \in [m]$, $2\tau \leq \abs{S_j} \leq
n/2$. Let $A \subseteq [n]$ be a set of size $n-4p$ such that for every
$j \in [m]$, $\abs{A\cap S_j} \leq 0.01\abs{S_j}$ and $A$ is disjoint
from all sets $S_i$ of size at most $10^4\tau$. Let $B$ be an arbitrary
subset of $A$. Then, there exists a set $T \subseteq [n]\setminus A$
of size exactly $3p$, such that for every $j \in [m]$,  if $|S_j| > 2\tau$ then
for every integer $t$ with $-\tau < t \leq \tau $, it holds that $\abs{\left(T\cup B\right)
\cap S_j} \neq  \floor{\abs{S_j}/2} + t  \bmod p$. If $|S_j| = 2\tau$, the same holds
for $-\tau < t < \tau$.
\end{lemma}
\begin{proof}
Denote $\widetilde{[n]} = [n] \setminus A$, and $\tilde{S}_i = S_i
\setminus A$ for all $i \in [m]$. We note that if $\abs{S_i} \leq
10^4 \tau$, then $\tilde{S}_i = S_i$. We construct the set $T$ by a
randomized algorithm, which consists of several steps. In the first
step, we greedily select a small number of elements from each set
$\tilde{S}_i$. The purpose of this step is to guarantee that $|T \cap
S_i|$ is sufficiently far from $0$, for every $i$. Next, we pick each
of the remaining elements of $\widetilde{[n]}$ to $T$ with probability
$0.65$. This constant is chosen so that with high probability (assuming
$|\tilde{S}_i|$ is sufficiently large), the intersection $|T \cap
\tilde{S}_i|$ is non-zero modulo $p$ (and since $|\tilde{S}_i|$ and
$|S_i|$ are very close, the same holds for $|T \cap S_i|$), and also
with high probability the number of elements we have picked so far does
not exceed $3p$.

The next step is again a deterministic, greedy step, which adds to $T$
sufficiently many elements from each ``bad'' set $S_i$. Those are the sets
of which too few elements were picked before. By standard concentration
bounds, we do not expect to have many such large sets, and thus again
we can control the number of elements added in this step.

Finally, assuming the number of elements that were picked so far is
less than $3p$ (which happens with high probability), we add arbitrary
elements to our set so that it will be of size exactly $3p$. Of course,
we also have to argue that this step preserves the previous intersection
requirements. This follows from the fact that we do not expect to add
many elements in this step.

We now provide the more formal details. $T$ is constructed using the
following randomized algorithm.

\begin{itemize}
\item 
For every $j\in [m]$ such that $\abs{S_j} \leq 6000\tau$, we add
all elements of $S_j$ to $T_1$. We then take $6000\tau$ arbitrary
elements from the remaining  sets among $\tilde{S}_1, \tilde{S}_2, \ldots,
\tilde{S}_m$. Since $m \leq 10^{-5}n/\tau$, the size of $T_1$ is at most
$0.06 n$. Without loss of generality, we take $T_1$ to be of size equal
to $0.06 n$.
\item 
Let $T_2$ be the set obtained by picking every element in
$\widetilde{[n]}\setminus T_1$  independently with probability $0.65$.
\item 
For every $j \in [m]$, such that $\abs{S_j \cap \left(T_1\cup T_2 \cup
B\right)} \leq 0.52\abs{S_j}$, include all elements in $\tilde{S}_j
\setminus \left(T_1\cup T_2\right)$ in the set $T_3$.
\item 
If $\abs{T_1} + \abs{T_2} + \abs{T_3} > 3p$, abort. Else, we add
$3p-\abs{T_1}- \abs{T_2}  -\abs{T_3}$  arbitrary elements from
$\widetilde{[n]}\setminus \left(  T_1 \cup T_2 \cup T_3\right)$ into
the set $T_4$.
\item 
Let $T = T_1 \cup T_2 \cup T_3 \cup T_4$. 
\end{itemize}
We will now argue that with a high probability, the algorithm above
outputs a set $T$ which satisfies the desired properties. To this end,
we need the following claims, whose proofs we defer to the end of this
section. The probabilities in these claims are all taken over the choice
of $T_2$, which is the only randomized step in the algorithm.
\begin{claim}\label{clm : general n, T2}
With probability at least $1-n^{-5}$, all of the following events happen.
\begin{itemize}
\item $0.64 n \leq \abs{T_2} \leq 0.66 n \, .$
\item $\forall j \in [m], \text{ such that } \abs{S_j}\geq  1000\log n$,  
${\abs{\tilde{S}_j \cap T_2} \in  [0.52 \abs{S_j}}, 0.74 \abs{S_j}] $.
\item For every $j \in [m]$, 
if $\abs{S_j} \leq6000\tau$, then $S_j \subseteq T$.
\item For every $j \in [m]$, if $\abs{S_j} \geq6000\tau$, 
then $\abs{S_j\cap T} \geq \max\{{6000\tau, 0.52\abs{S_j}}\}$. 
\end{itemize}
\end{claim}
\begin{claim}[$T_3$ is typically small]\label{clm : general n, T3}
\[
\Pr[\abs{T_3} \leq 0.01 n] \geq 0.99 \, .
\]
\end{claim}
\begin{claim}[$T_4$ is typically small]\label{clm : general n, T4} 
\[
\Pr[\abs{T_4} \leq 0.05 n] \geq 1-n^{-5} \, .
\]
\end{claim}
\paragraph*{Probability of aborting and size of $T$. } The algorithm
aborts only in the case that $\abs{T_1} + \abs{T_2} + \abs{T_3} >
3p$. We know that with probability $1$, $\abs{T_1} \leq 0.06 n$. It
follows from~\autoref{clm : general n, T2} that with probability at least
$1-n^{-5}$, $\abs{T_2} \leq 0.66 n$ and from~\autoref{clm : general n,
T3} that with probability at least $0.99$, $\abs{T_3} \leq 0.01 n$. Thus,
with probability at least $0.98$, $\abs{T_1} + \abs{T_2} + \abs{T_3}
\leq 0.73 n$. Since $4p \leq n \leq 4p + O(p^{0.6})$,  with probability
at least $0.98$, $\abs{T_1} + \abs{T_2} + \abs{T_3} \leq 3p$. Also,
whenever the algorithm does not abort, the set $T_4$ is picked so that
$T$ output by the algorithm satisfies $\abs{T} = 3p$.
\paragraph*{Intersection properties of $T$. }
For the rest of this argument, we assume that $T_1, T_2, T_3, T_4$ satisfy
the properties in~\autoref{clm : general n, T2},~\autoref{clm : general n,
T3} and~\autoref{clm : general n, T4}.  We now argue that for every $j
\in [m]$ it holds that $\abs{\left(T\cup B\right)
\cap S_j} \neq  \floor{\abs{S_j}/2} + t  \bmod p$ for every integer $t$ in the
range specified in the statement of the Lemma.

We consider some cases based on the size of $S_j$. 
\begin{itemize}
\item {\bf Very small sets : $2\tau \leq \abs{S_j} \leq 6000\tau$. }From~\autoref{clm : general n, T2}, all such sets are completely contained in $T$. Thus,
\[
 \abs{\left(T\cup B\right) \cap S_j} -  (\floor{\abs{S_j}/2} + t) = \lceil{\abs{S_j}/2}\rceil - t \, .
\]
Since $1\leq \tau \leq p/10^5$, this remains non-zero modulo $p$ for every $-\tau < t \le \tau $ if $|S_j| > 2\tau$, and for every $-\tau < t < \tau$ if $|S_j|=2\tau$.
\item {\bf Small sets : $6000\tau < \abs{S_j} \leq 10^4 \tau$. } From~\autoref{clm : general n, T2}, we know that for every $j \in [m]$, $\abs{S_j \cap T} \geq 6000\tau$. We get that for every $-\tau < t \le \tau $, 
\[
 1 \leq \abs{\left(T\cup B\right) \cap S_j} -  (\floor{\abs{S_j}/2} + t) \leq (10^4 + 1)\tau
\]
Since $\tau \leq  p/10^5$, $\abs{\left(T\cup B\right) \cap S_j} -  (\floor{\abs{S_j}/2} + t) $ is non-zero modulo $p$ for each $-\tau < t \le \tau $.

\item {\bf Sets of intermediate size : $10^4\tau < \abs{S_j} \leq  1000\log n$. }Since by \autoref{clm : general n, T2}, $\abs{S_j \cap T} \geq 0.52\abs{S_j}$, we get that for every $-\tau < t \le \tau $,
\[
198 \tau \leq \abs{\left(T\cup B\right) \cap S_j} 
-  (\floor{\abs{S_j}/2} + t) \leq  1000 \log n.
\]
Thus, $\abs{\left(T\cup B\right) \cap S_j} -  (\floor{\abs{S_j}/2} + t)$ remains non-zero modulo $p$.  

\item {\bf Large sets : $ \max\{1000\log n, 10^4\tau\} \leq \abs{S_j} \leq n/2   $. }For such large sets, from~\autoref{clm : general n, T2}, \autoref{clm : general n, T3} and~\autoref{clm : general n, T4}, we know that
\begin{align*}
0.52 \abs{S_j} \leq \abs{\left(T\cup B\right) \cap S_j}  & = \sum_{k=1}^4 \abs{T_k \cap S_j} + |B \cap S_j| \\
&  \leq 0.74\abs{S_j} + \abs{T_1} + \abs{T_3} +\abs{T_4} + 0.01 \abs{S_j} \le 0.75\abs{S_j} + 0.12n,
\end{align*}
where we have also used the assumption that $|A \cap S_j| \le 0.01\abs{S_j}$, which in particular implies this upper bound for $|B \cap S_j|$, as $B \subseteq A$.
Thus, as $\abs{t} \le \tau \le 10^{-4}\abs{S_j}$,
\[
0.02\abs{S_j} - \tau \leq \abs{\left(T\cup B\right) \cap S_j} 
-  (\floor{\abs{S_j}/2} + t) \leq 0.251 \abs{S_j} + 0.12 n \, .
\]
Using $\abs{S_j} \leq n/2$, $4p + n^{0.6} \geq n$ and $\abs{S_j} \geq 10^4\tau$  we get that
\[
0 < \abs{\left(T\cup B\right) \cap S_j} -  (\floor{\abs{S_j}/2} + t) \leq 0.99 p \, .
\]
So, this quantity is also non-zero modulo $p$.
\end{itemize}
These three cases complete the proof of the lemma.
\end{proof}

We can now prove \autoref{thm:unbalancing-general}.

\begin{proof}[Proof of \autoref{thm:unbalancing-general}]
We follow the outline discussed at the beginning of this section.
Without loss of generality, we can assume that each set $S_i$ has size
at most $n/2$, else we work with the complement of $S_i$. Suppose, for
the sake of contradiction, that $m \leq \frac{1}{10^5}\cdot n/\tau$. Let
$p$ be the largest prime such that $4p \le n$. For large enough $n$,
there is such a prime $p$ such that $n-4p \leq n^{0.6}$ (see~\cite{BHP01}).

Let $A \subseteq [n]$ be the set of size $n-4p \le n^{0.6}$ given by
\autoref{lem : reducing to 4p}. Let $B$ be an arbitrary subset of $A$
of size $|A|/2$.

To every element $k \in [n]\setminus A$, we associate a formal variable $x_k$, and let $\vx = \{ x_k : k \in [n]\setminus A\}$ (note that $|\vx| = 4p$).
For each $j \in [m]$ such that $|S_j| > 2\tau$, define the following polynomials over $\F_p$:
\[
B_j (\vx) = \prod_{t=-\tau + 1}^{\tau} \left(\sum_{k \in S_j\setminus A} 
x_{k} + \abs{S_j \cap B} - \floor{|S_j|/2} - t\right)\, .
\]
If $|S_j|=2\tau$, define a similar polynomial $B_j$ where $t$ ranges from $\tau + 1$ to $\tau - 1$.
Further, let
\[
f(\vx) = \prod_{j=1}^m B_j (\vx)\, ,
\]
be a polynomial over $\F_p$. From the choice of the 
set $A$ (see \autoref{lem
: reducing to 4p}), we know that for every $j \in [m]$, $B_j$ is a
non-zero polynomial of degree smaller than $2 \tau$.

There is a natural bijection between $[n] \setminus A$ and $[4p]$ (say,
by ordering the elements of $[n]\setminus A$ by increasing order). Thus,
we can naturally associate subsets $Y'$ of $[4p]$ with subsets of
$[n] \setminus A$, and indicator vector $\ind_{Y'}$ with elements of
$\set{0,1}^{\vx}$.

We would like first to argue that $f$ vanishes over all vectors of the
form $\ind_{Y'}$ for $Y' \in \binom{[4p]}{2p}$. Indeed, let $Y'$ be such
a set, and extend it to a balanced partition of $[n]$ by considering $Y =
Y' \cup B$.

By the assumption, there is an index $j$ such that
$|Y \cap S_j| - \floor{S_j}/2 \in \set{-\tau+1,
\ldots, \tau}$, and since $|Y \cap S_j| = |Y' \cap S_j| + |B \cap
S_j|$, it follows that $B_j(\ind_{Y'}) = 0$ and thus $f(\ind_{Y'}) =
0$, as required.

Next, we want to show $f$ does not vanish over a vector $\ind_T$ for
some $T \in \binom{[4p]}{3p}$.

Indeed, \autoref{lem:3p general} precisely guarantees the existence
of such a set $T \subseteq [n]\setminus A$, of size equal to $3p$, so
that for all $j \in [m]$, $B_j(\ind_T) \not\equiv 0 \bmod p$, and thus
$f(\ind_T) \neq 0$..

By~\autoref{lem:hegedus}, $\deg(f) \ge p$, and by construction, $\deg(f)
\le  2 \tau \cdot m$, contradicting the assumed lower bound on $m$.
\end{proof}
%We now prove~\autoref{lem : reducing to 4p} and~\autoref{lem:3p general}.

\subsubsection*{Proofs of~\autoref{clm : general n, T2},~\autoref{clm : general n, T3} and~\autoref{clm : general n, T4}}
We now prove the claims needed in the proof of~\autoref{lem:3p general}. The arguments are based on standard concentration bounds.

\begin{proof}[Proof of~\autoref{clm : general n, T2}]
The expected size of the set $T_2$ is equal to  $0.65 \abs{[n]\setminus A}$. Using the fact that $\abs{A} \leq n^{0.6}$ and by~\autoref{lem : hoeffding 1},  we get that with probability at least $1-\exp({-\Omega(n)})$, 
\[
0.64 n \leq \abs{T_2} \leq 0.66 n \, .
\]
For the second item, observe that for any fixed  $j \in [m]$, by~\autoref{lem : hoeffding 1}, we have 
\[
\Pr\left[\abs{\abs{\tilde{S}_j \cap T_2} - 0.65 \abs{\tilde{S_j}}} \geq 0.09 \abs{\tilde{S_j}}\right] \leq 2\exp\left(-0.0162\abs{\tilde{S_j}}\right) \, .
\]
We know that $\tilde{S}_j \subseteq {S}_j$ and $\abs{\tilde{S}_j} \geq 0.99\abs{S}$. 
%From our guarantees on $A$, we know that $\abs{A\cap S_j} \leq 0.01\abs{S_j}$. 
Thus, 
\[
\Pr\left[{\abs{\tilde{S}_j \cap T_2} \in  [0.52 \abs{S_j}}, 0.74 \abs{S_j}]\right] \geq 1-2\exp\left(-0.015\abs{S_j}\right)\, .
\] 
For sets $S_j$ of size at least $1000\log n$, this probability is high enough to take a union bound over all sets. So, we have the following. 
\[
\Pr\Big[\forall j \in [m] \text{ such that } \abs{S_j}\geq 1000\log n, {\abs{\tilde{S}_j \cap T_2} \in  [0.52 \abs{S_j}}, 0.74 \abs{S_j}]\Big] \geq 1- n^{-8}\, .
\] 

For the third and fourth items,  observe that by construction, the
set $T_1$ is a superset of all sets of size at most $6000\tau$ and
intersects every $S_j$ on at least $6000\tau$ elements. Moreover, since
$\abs{\tilde{S}_j} \geq 0.99\abs{S_j}$,  it follows that if $\abs{S_j \cap
(T_1\cup T_2)} \leq 0.52\abs{S_j}$, then sufficiently many elements will
be included in the set $T_3$ so that $\abs{S_j \cap (T_1\cup T_2 \cup
T_3)} \geq 0.52\abs{S_j}$.  \end{proof}

\begin{proof}[Proof of~\autoref{clm : general n, T3}]
For $j \in [m]$, we say that the set $S$ is violated if $\abs{S_j \cap
(T_1 \cup T_2 \cup B)} \leq 0.52\abs{S_j}$. Since $T_1$ intersects every
set $S_j$ on at least $6000\tau$ elements, we know that any violated
set $S_j$ must satisfy $\abs{S_j} \geq 10^4\tau$.
So, from the proof of~\autoref{clm : general n, T2}, we get that the
expected size of the set $T_3$ is given by
\[
\E[\abs{T_3}] \leq \sum_{j \in [m], \abs{S_j} \geq 10^4\tau} \frac{2\abs{S_i}}{\exp(0.015\abs{S_i})} \, .
\]
From~\autoref{claim : decreasing fun} below, we know that this expectation can be upper bounded by
\[
\E[\abs{T_3}] \leq m \cdot  \frac{2\cdot \abs{10^4\tau}}{\exp(0.015 \times 10^4\tau)} \, .
\]
Since $\tau$ is at least $1$ and $m \leq n/\tau$, we get
\[
\E[\abs{T_3}] \leq 10^{-10} n \, .
\]
By Markov's inequality, we get the claim.
\end{proof}

\begin{proof}[Proof of~\autoref{clm : general n, T4}]
This immediately follows from~\autoref{clm : general n, T2}. Observe that
\[
\abs{T_4} \leq 3p - \abs{T_1} - \abs{T_2} \, .
\]
$\abs{T_2}  \geq 0.64 n$ with probability at least $1-n^{-5}$, and $\abs{T_1} \geq 0.06 n$ with probability $1$. Thus, with probability at least $1-n^{-5}$, $\abs{T_4} \leq 0.05 n$.
\end{proof}

\begin{claim}\label{claim : decreasing fun}
Let $c$ be any positive constant. Then, for any $y \geq x \geq 1/c$, it holds that $x\cdot \me^{-cx} \geq y\cdot \me^{-cy}$.
\end{claim}
\begin{proof}
Let $f(x)= x\cdot \me^{-cx}$.  The first derivative of $f(x)$ is
\[
f'(x) = \me^{-cx} -cx\me^{-cx} \, .
\] 
It is easy to see that this is positive for $0<x\leq 1/c$ and
negative for $x >1/c$. Therefore, $f(x)$, which vanishes at $0$,
increases as $x$ increases from $0$ to $1/c$, achieves
its maximum at $x = 1/c$ and decreases thereafter. This implies the claim.
\end{proof}

\section{Syntactically Multilinear Arithmetic Circuits}\label{sec:proof of main theorem}
In this section, for the sake of completeness, we review the arguments of Raz, Shpilka and Yehudayoff \cite{RSY08}, and show how \autoref{thm:unbalancing} implies a lower bound of $\Omega(n^2/\log^2 n)$. We mostly refer for \cite{RSY08} for the proofs.

Specifically, we will show the following.
\begin{theorem}
\label{thm:lower-bound}
Let $n$ be an even integer, and $X=\set{x_1, \ldots, x_n}$. Let $f(X) \in \F[X]$ be a multilinear polynomial such that for every balanced partition $X=Y \sqcup Z$, $\rank_{Y,Z}(f) = 2^{n/2}$. Let $\Psi$ be a syntactically multilinear circuit computing $f$. Then $|\Psi| = \Omega(n^2/\log^2 n)$.
\end{theorem}

The first step in proof of \autoref{thm:lower-bound} is to show
 that if $f$ is computed by a syntactically mutilinear circuit of size
 $s$, then there exists a syntactically multilinear circuit of size
 $O(s)$ that computes all the first-order partial derivatives of $f$,
 with the additional important property that for each $i$, the variable
 $x_i$ does not appear in the subcircuit rooted at the output gate which
 computes $\partial f / \partial x_i$.

\begin{theorem}[\cite{RSY08}, Theorem 3.1] \label{thm:BS-multilinear}
Let $\Psi$ be a syntactically multilinear circuit over a field $\F$
and the set of variables $X=\set{x_1, \ldots, x_n}$. Then, there
exists a syntactically multilinear circuit $\Psi'$, over $\F$ and
$X$, such that: \begin{enumerate} \item $\Psi'$ computes all $n$
first-order partial derivatives $\partial f / \partial x_i$, $i \in
[n]$.  \item \label{item:BS-size} $|\Psi'| \le 5|\Psi|$.  \item $\Psi'$
is syntactically multilinear.  \item \label{item:key-property}For every
$i \in [n]$, $x_i \not \in X_{v_i}$, where $v_i$ is the gate in $\Psi'$
computing $\partial f / \partial x_i$.  \end{enumerate} In particular,
if $v$ is a gate in $\Psi'$, then it is connected by a directed path to
at most $n-|X_v|$ output gates.  \end{theorem}

The proof of \autoref{thm:BS-multilinear} appears in \cite{RSY08}, and mostly follows the classical proof of Baur and Strassen \cite{BS83} of the analogous result for general circuits, with additional care in order to guarantee the last two properties.

Next we define two types of gates in a syntactically multilinear arithmetic circuits.

\begin{definition}
\label{def:upper-lower-levels}
Let $\Phi$ be a syntactically multilinear arithmetic circuit. Define $\cL(\Phi, k)$, the set of lower-leveled gates in $\Phi$, by
\[
\cL(\Phi,k) = \setdef{u}{\text{$u$ is a gate in $\Phi$, $k < |X_u| < n-k$, and $u$ has a parent $v$ with $|X_{v}| \ge n-k$}}.
\]
Define $\cU(\Phi,k)$, the set of upper-leveled gates in $\Phi$, by
\[
\cU(\Phi,k) = \setdef{v}{\text{$v$ is a gate in $\Phi$, $|X_v| \ge n-k$, and $u$ has a child $v \in \cL(\Phi,k)$}}. \qedhere
\]
\end{definition}

The following lemma shows that if the set of lower-leveled gates is small, then there exists a partition $X=Y\sqcup Z$ under which the polynomial computed by the circuit is not of full rank.

\begin{lemma}
\label{lem:small-lower-leveled-not-full-rank}
Let $\Phi$ be a syntactically multilinear arithmetic circuit over $\F$ and $X=\set{x_1, \ldots, x_n}$, for an even integer $n$, computing $f$. Let $\tau=3\log n$ and $\cL = \cL(\Phi, 100\tau)$. If $|\cL| < n / (10^5 \tau)$, then there exists a partition $X= Y \sqcup Z$ such that $\rank_{Y,Z}(f) < 2^{n/2-1}$.
\end{lemma}

We first sketch how \autoref{thm:lower-bound} follows from \autoref{lem:small-lower-leveled-not-full-rank}. The proof is identical to the proof given in \cite{RSY08} with slightly different parameters.

\begin{proof}[Proof of \autoref{thm:lower-bound} assuming \autoref{lem:small-lower-leveled-not-full-rank}]
Let $\Psi'$ be the arithmetic circuit computing all $n$ first-order partial derivatives of $f$, given by \autoref{thm:BS-multilinear}. Set $\tau = 3\log n$ and let $\cL = \cL(\Psi', 100\tau)$ and $\cU = \cU(\Psi', 100\tau)$ as in \autoref{def:upper-lower-levels}.

Denote $f_i = \partial f / \partial x_i$ and let $v_i$ be the gate in $\Psi'$ computing $f_i$, and $\Psi'_i$ be the subcircuit of $\Psi'$ rooted at $v_i$. Let $\cL_i = \cL(\Psi'_i, 100\tau)$. It is not hard to show (see \cite{RSY08}) that $\cL_i \subseteq \cL$, and by \autoref{lem:small-lower-leveled-not-full-rank} and \autoref{item:full-rank-derivative} in \autoref{prop:pdmatrix}, it follows that $|\cL_i| \ge n/(10^5\tau)$.

For every gate $v$ in $\Psi'$ define $C_v = \setdef{i \in [n]}{\text{$v$ is a gate in $\Psi_i$}}$ to be the set of indices $i$ such that there exists a directed path from $v$ to the output gate computing $f_i$. For $i \in [n]$, let $\cU_i = \setdef{u \in \cU}{\text{$u$ is a gate in $\Psi'_i$}}$, so that $\sum_{u \in \cU} C_u = \sum_{i \in [n]} |\cU_i|$.

Since the fan-in of each gate is at most two, $|\cL_i| \le 2|\cU_i|$, and since every $u \in \cU$ satisfies $|X_u| \ge n-100\tau$, it follows by \autoref{thm:BS-multilinear} that $|C_u| \le 100 \tau$. Thus, we get
\[
n \cdot \frac{n}{10^5\tau} \le \sum_{i \in [n]} |\cL_i| \le 2 \sum_{i \in [n]} |\cU_i| = 2\sum_{u \in U} C_u \le 2|\cU| \cdot 100\tau.
\]
By \autoref{item:BS-size} in \autoref{thm:BS-multilinear}, and $\tau=3 \log n$,
\[
|\Psi| = \Omega(|\Psi'|) = \Omega(|\cU|) = \Omega\left(\frac{n^2}{\log^2 n}\right). \qedhere
\]
\end{proof}

It remains to prove \autoref{lem:small-lower-leveled-not-full-rank}. As the proof mostly appears in \cite{RSY08}, we only sketch the main steps.

\begin{proof}[Proof sketch of \autoref{lem:small-lower-leveled-not-full-rank}]
Suppose $\cL \le n/(10^5\tau)$. By applying \autoref{thm:unbalancing-general} to the family of sets $\setdef{X_v}{v \in \cL}$, it follows that there exists a balanced partition $Y \sqcup Z$ of $X$ such that $X_v$ is $\tau$-unbalanced for every gate $v \in \cL$ (one could get slightly improved constants in the case $n=4p$ by applying \autoref{thm:unbalancing}).

The proof now proceeds in the exact same manner as the proof of Lemma 5.2 in \cite{RSY08}. In Proposition 5.5 of \cite{RSY08}, it is shown that one can write
\[
f = \sum_{i \in [\ell]} g_i h_i + g,
\]
where $\cL = \set{v_1, \ldots, v_{\ell}}$, $h_i$ is the polynomial computed at $v_i$, and the set of variables appearing in $g_i$ is disjoint from $X_{v_i}$.

In Claim 5.7 of \cite{RSY08}, it is shown that for every $i \in [\ell]$, $\rank_{Y,Z} (g_i h_i) \le 2^{n/2-\tau}$. This uses the fact that $X_{v_i}$ is $\tau$-unbalanced, the upper bound in \autoref{item:trivial-upper-bound} in \autoref{prop:pdmatrix}, and \autoref{item:rank-multiplicative} in the same proposition.

In Proposition 5.8 of \cite{RSY08}, it is shown (with the necessary change of parameters) that the degree of $g$ is at most $200\tau$.

Thus, by the fact that $\tau=3\log n$, \autoref{item:low-degree-rank} and \autoref{item:rank-additive} of \autoref{prop:pdmatrix}, it follows that for large enough $n$,
\[
\rank_{Y,Z} (f) \le \ell \cdot 2^{n/2-\tau} + 2^{\tau^3} < 2^{n/2-1}. \qedhere
\]
\end{proof}

\subsection{An explicit full-rank polynomial}

In this section, for the sake of completeness, we give a construction of a polynomial which is full-rank under any partition of the variables. 

\begin{construction}[Full rank polynomial, \cite{RSY08}]
\label{con:full-rank}
Let $n$ be an even integer, and let $\cW = \set{\omega_1, \ldots, \omega_{n}}$ and $X = \set{x_1, \ldots, x_n}$ be sets of variables.
For a set $B \in \binom{[n]}{n/2}$, denote by $i_1 < \cdots <  i_{n/2}$ the elements of $B$ in increasing order, and by $j_1 < \cdots < j_{n/2}$ the elements of $[n] \setminus B$ in increasing order.
Define $r_B = \prod_{\ell \in B} \omega_\ell$, and $g_B = \prod_{\ell \in [n/2]} (x_{i_\ell} + x_{j_\ell})$.

Finally, define \[
f= \sum_{B \in \binom{[n]}{n/2}} r_B g_B. \qedhere
\]
\end{construction}

\begin{claim}[\cite{RSY08}]
\label{cl:full-rank}
For $f$ from \autoref{con:full-rank}, it holds that for every balanced partition of $X = Y \sqcup Z$,
$\rank_{Y,Z} (f) = 2^{n/2}$, where the rank is taken over $\F(\cW)$.
\end{claim}

We give a proof which is shorter and simpler than the one given in \cite{RSY08}.

\begin{proof}[Proof of \autoref{cl:full-rank}]
Fix a balanced partition $X=Y \sqcup Z$, and consider the matrix $M_{Y,Z}(f)$ where $f$ is interpreted as a polynomial in $f \in (\cF[\cW])[X]$ (that is, the rows and columns of the matrix are indexed by $X$ variables and its entries are polynomials in $\cW$). We want to show that $\det(M_{Y,Z}(f)) \in \F[\cW]$ is a non-zero polynomial. Fix $\omega_i = 1$ if $i \in Y$ and $\omega_i = 0$ otherwise. Under this restriction, $f=g_Y$. It is also not hard to see that $\det(M_{Y,Z}(g_Y)) \neq 0$, since this is a permutation matrix (this also follows from \autoref{item:rank-multiplicative} of \autoref{prop:pdmatrix}).  Thus, $\det(M_{Y,Z}(f))$ evaluates to a non-zero value under this setting of the variables $\cW$, which implies it a non-zero polynomial.
\end{proof}

\begin{corollary}
Every syntactically multilinear circuit computing $f$ has size at least $\Omega(n^2 / \log^2 n)$.
\end{corollary}

The polynomial $f$ in \autoref{con:full-rank} is in the class $\VNP$
of explicit polynomials, but it is not known whether there exists a
polynomial size multilinear circuit for $f$.

Raz and Yehudayoff \cite{RY08} constructed a full-rank polynomial $g
\in \F[X,\cW']$ that has a  syntactically multilinear circuit of size
$O(n^3)$. Their construction also uses a set of auxiliary variables $\cW'$
of size $O(n^3)$. Thus, if one measures the complexity as a function of
$|X| \cup |\cW'|$, the quadratic lower bound of \autoref{thm:lower-bound}
is meaningless, because a lower bound of $\Omega(n^3)$ holds
trivially. However, we believe that since the rank is taken over
$\F(\cW')$, it is only fair to consider computations over $\F(\cW')$,
where any rational expression in the variables of $\cW'$ is merely a
field constant. Thus, in this setting, an input gate can be labeled by
an arbitrarily complex rational function in the variables of $\cW'$,
and the complexity is measured as a function of $|X|$ alone. In this
model the lower bound of \autoref{thm:lower-bound} \emph{is} meaningful,
and furthermore, this example shows that the partial derivative matrix
technique cannot prove an $\omega(n^3)$ lower bound.

\section*{Acknowledgments}
Part of this work was done while the first author was visiting  Tel Aviv
University. We thank  Amir Shpilka for  the visit, for many insightful
discussions, and for comments on an earlier version of this text. We
are also thankful to Andy Drucker for pointing out a correction in a
previous version of this paper.

\bibliographystyle{customurlbst/alphaurlpp}
\bibliography{references}

\end{document}

%% file: thmmacros.tex
\usepackage{amsthm}
\usepackage{thmtools,thm-restate}

\numberwithin{equation}{section}
\declaretheoremstyle[bodyfont=\it,qed=\qedsymbol]{noproofstyle}

\declaretheorem[name=Observation,numbered=no]{observation*}

\declaretheorem[numberlike=equation]{theorem}

\declaretheorem[name=Theorem,numbered=no]{theorem*}

\declaretheorem[numberlike=equation]{lemma}
\declaretheorem[name=Lemma,numbered=no]{lemma*}

\declaretheorem[numberlike=equation]{corollary}
\declaretheorem[name=Corollary,numbered=no]{corollary*}

\declaretheorem[numberlike=equation]{proposition}
\declaretheorem[name=Proposition,numbered=no]{proposition*}

\declaretheorem[numberlike=equation]{claim}
\declaretheorem[name=Claim,numbered=no]{claim*}

\declaretheorem[name=Conjecture,numbered=no]{conjecture*}

\declaretheorem[numberlike=equation]{question}
\declaretheorem[name=Question,numbered=no]{question*}

\declaretheoremstyle[bodyfont=\it,qed=$\lozenge$]{defstyle} 

\declaretheorem[numberlike=equation,style=defstyle]{definition}
\declaretheorem[unnumbered,name=Definition,style=defstyle]{definition*}

\declaretheorem[unnumbered,name=Example,style=defstyle]{example*}

\declaretheorem[unnumbered,name=Notation=defstyle]{notation*}

\declaretheorem[numberlike=equation,style=defstyle]{construction}
\declaretheorem[unnumbered,name=Construction,style=defstyle]{construction*}

\declaretheorem[unnumbered,name=Remark,style=defstyle]{remark*}

%% file: customurlbst/bibmacros.tex
\usepackage{nth}
\usepackage{intcalc}
\usepackage{etoolbox}
\usepackage{xstring}
\hypersetup{
}

\usepackage{ifpdf}
\ifpdf
\else
\usepackage[quadpoints=false]{hypdvips}
\fi

\newcommand{\ECCCurl}[2]{http://eccc.weizmann.ac.il/report/\ifnumcomp{#1}{>}{93}{19}{20}#1/#2/}

\newcommand{\shortECCC}[2]{\texttt{\href{http://eccc.weizmann.ac.il/report/\ifnumcomp{#1}{>}{93}{19}{20}#1/#2/}{eccc:TR#1-#2}}}

\newcommand{\parseECCC}[1]{% Takes a string of the form TRxx/xxx or
\StrSubstitute{#1}{TR}{}[\tmpstring]%
\IfSubStr{\tmpstring}{/}{ %assuming string is of the form TRxx/xxx
\StrBefore{\tmpstring}{/}[\ecccyear]%
\StrBehind{\tmpstring}{/}[\ecccreport]%
}{% assuming string is of the form TRxx-xxx
\StrBefore{\tmpstring}{-}[\ecccyear]%
\StrBehind{\tmpstring}{-}[\ecccreport]%
}%
\shortECCC{\ecccyear}{\ecccreport}}